\documentclass[a4paper,11pt]{article}
\usepackage{graphicx}
\usepackage{latexsym}
\usepackage{amssymb}
\usepackage{amsfonts}
\usepackage{amsmath}
\usepackage{amsthm}
\usepackage{indentfirst}

\newtheorem{proposition}{Proposition}[section]

\def\first{{ \mathsf{first}}}
\def\last{{ \mathsf{last}}}
\newtheorem{Rem}{Remark}

\author{A. Bernini\thanks{Dipartimento di Matematica e Informatica ``U. Dini'', Universit\`a degli Studi di Firenze, Viale
 G.B. Morgagni 65, 50134 Firenze, Italy. {
 \tt \ antonio.bernini@unifi.it,\quad stefano.bilotta@unifi.it,\quad renzo.pinzani@unifi.it}}
 \and S. Bilotta$^*$\and R. Pinzani$^*$\and V. Vajnovszki\thanks{LE2I, Universit\'e de Bourgogne, BP 47 870, 21078 Dijon Cedex, 
 France{ \tt \ vvajnov@u-bourgogne.fr}}}

\title{A Gray Code for cross-bifix-free sets}

\begin{document}

\maketitle

\begin{abstract}
A cross-bifix-free set of words is a set in which no prefix of any length of any
word is the suffix of any other word in the set. A construction of
cross-bifix-free sets has recently been proposed by Chee {\it et al.} in 2013
within a constant factor of optimality. We propose a \emph{trace partitioned} Gray code for these cross-bifix-free sets
and a CAT algorithm generating it.
\end{abstract}

\section{Introduction}
A cross-bifix-free set of words is a set where, given any two words over an alphabet,
possibly the same, any prefix of the first one is not a suffix of the second one
and {\it vice versa}. Cross-bifix-free sets are involved in the study of distributed sequences
for frame synchronization \cite{de lind}. The problem of determining such sets is also related
to several other scientific applications, for instance in pattern
matching \cite{croche} and automata theory \cite{berstel}.

Fixed the cardinality $q$ of the alphabet and the length $n$ of the
words, a matter is the construction of a cross-bifix-free set with the cardinality
as large as possible. An interesting method has been proposed in \cite{Bajic} for words over
a binary alphabet. In a recent paper \cite{singa} the authors revisit
the construction of \cite{Bajic} and generalize it obtaining cross-bifix-free sets of words
with greater cardinality over an alphabet of arbitrary size. They also show that their cross-bifix-free
sets have a cardinality close to the maximum possible; and to our knowledge this is the best result 
in literature about the size of cross-bifix-free sets.

It is worth to mention that an intermediate step between
the original method \cite{Bajic} and its generalization in \cite{singa} has been proposed in 
\cite{cross}: it is constituted by
a different construction of binary cross-bifix-free sets based on lattice paths
which allows to obtain greater cardinality if compared to the ones in \cite{Bajic}.

\medskip

Once a class of objects is defined, in our case words, often it could be useful to list or generate
them according to a particular criterion. A special way to do this is their generation
in a way such that any two consecutive words differ as little as
possible, i.e., in Gray code order \cite{G}. In the case the objects are words, as in our, we can
specialize the concept of Gray code saying that it is \emph{an infinite set of word-lists
with unbounded word-length such that the Hamming distance between any two
adjacent words is bounded independently of the word-length} \cite{Wa1} (the Hamming
distance is the number of positions in which the two successive words differ \cite{H}). Gray
codes find useful applications in circuit testing, signal encoding, data compression,
telegraphy, error correction in digital communication and others. They are
also widely studied in the context of combinatorial
objects as: permutations \cite{J},
Motzkin and Schr\"oder words \cite{V2}, derangements \cite{BV},
involutions \cite{Wa2}, compositions, combinations, set-partitions \cite{R,Sa}, and so on.

\medskip

In this work we propose a Gray code for
the cross-bifix-free set $S_{n,q}^{(k)}$ defined in \cite{singa}. It is formed by
length $n$ words over the $q$-ary alphabet $A=\{0,1,\ldots,q-1\}$ containing a particular sub-word avoiding $k$
consecutive 0's (for more details see the next section). First we propose a Gray code for $S_{n,2}^{(k)}$
over the binary alphabet $\{0,1\}$, then we expand each binary word to the alphabet $A$.
The expansion of a binary word $\alpha$ is obtained replacing all the $1$'s with the symbols of $A$ different
from 0 producing a set of words with the same \emph{trace} $\alpha$. The Gray code we get is
\emph{trace partitioned} in the sense that all the words with the same trace are consecutive.

\section{Definitions and tools}

\noindent Let $n \ge 3$, $q \ge 2$ and $1 \le k \le n-2$. The cross-bifix-free set $S_{n,q}^{(k)}$
defined in \cite{singa} 
is the set of all length $n$ words $s_1 s_2 \cdots s_n$ over the
alphabet $\{0,\dots,q-1\}$ satisfying:
\begin{itemize}
\item $s_1 = \dots = s_k = 0$;
\item $s_{k+1} \ne 0$;
\item $s_n \ne 0$;
\item the subword $s_{k+2} \dots s_{n-1}$ does not contain $k$ consecutive 0's.
\end{itemize}

Throughout this paper we are going to use several standard notations
which are typical in the framework of  sets and lists
of words. For the sake of clearness we
summarize the ones used here.

For a set of words $L$ over an alphabet $A$ we denote by $\mathcal L$
an ordered list for $L$, and

\begin{itemize}

\item $\overline{\mathcal L}$ denotes the list obtained by covering $\mathcal L$ in reverse order;

\item if $\mathcal L'$ is another list, then $\mathcal L\circ\mathcal L'$
is the concatenation of the two lists, obtained by appending the words of
$\mathcal L'$ after those of $\mathcal L$;

\item $\first(\mathcal L)$ and $\last(\mathcal L)$ are the first and the last word
of $\mathcal L$, respectively;

\item if $u$ is a word in $A^*$, then $u\cdot\mathcal L$ (resp. $\mathcal L \cdot u$) is a new list where
each word has the form $u\omega$ (resp. $\omega u$) where $\omega$ is any word of $\mathcal L$;

\item if $u$ is a word in $A^*$, then $|u|$ is its length, and
$u^n = \underbrace{uuu\ldots u}_{n}$.
\end{itemize}


\noindent
For our purpose we need a Gray code list for the set of words of a certain length over the
$(q-1)$-ary
alphabet $\{1,2,\ldots,q-1\}$, $q\geq3$.
An obvious generalization of the Binary Reflected Gray Code
\cite{G} to the alphabet $\{1,2,\ldots,q-1\}$ is the list
$\mathcal{G}_{n,q}$ for the set of words $\{1,2,\ldots,q-1\}^n$
defined in \cite{E,Wi} where is also
shown that it is a Gray code with Hamming distance~$1$.
The authors defined this list as:

\begin{equation}
\label{g_nq}
\mathcal{G}_{n,q}=\left\{
\begin{array}{cr}
\lambda &\ \mathrm{if}\ n=0,\\
\\
1\cdot\mathcal{G}_{n-1,q}\circ
2\cdot\overline{\mathcal{G}_{n-1,q}}\circ
\cdots\circ
(q-1)\cdot\mathcal{G}_{n-1,q}' & \mathrm{if}\ n>0,
\end{array}
\right.
\end{equation}
where $\mathcal{G}_{n-1,q}'$ is $\mathcal{G}_{n-1,q}$ or
$\overline{\mathcal{G}_{n-1,q}}$
according on whether $q$ is even or odd.
The reader can easily
verify the following proposition.

\begin{proposition}\label{first_last}
For $q\geq 3$,
\begin{itemize}
\item $\first(\mathcal{G}_{n,q})=1^n$;
\item $\last(\mathcal{G}_{n,q})=(q-1)1^{n-1}$
      if $q$ is odd, and $(q-1)^n$ if $q$ is even.
\end{itemize}
\end{proposition}
\medskip

\noindent
Now we are going to present another tool we need in the paper.
If $\beta$ is a binary word of length $n$ such that $|\beta|_1=t$ (the number of 1's in
$\beta$), we define the \emph{expansion} of $\beta$, denoted by $\epsilon(\beta)$,
as the list of $(q-1)^t$ words, where the $i$-th word is obtained by replacing
the $t$ 1's of $\beta$ by the $t$ symbols (read from left to right) of the $i$-th word in
$\mathcal G_{t,q}$.
For example, if $q=3$ and $\beta=01011$ (the trace), then $\mathcal G_{3,3}=(111,112,122,121,221,222,212,211)$ and
$\epsilon(\beta)=(01011,01012,01022,01021,02021,02022,02012,02011).$
Notice that in particular $\first(\epsilon(\beta))=\beta$ and all the words of
$\epsilon(\beta)$ have the same trace.

We observe that $\epsilon(\beta)$ is the list obtained from
$\mathcal G_{t,q}$ inserting
some 0's, each time in the same positions.
Since $\mathcal G_{t,q}$
is a Gray code and the insertions of the 0's does not change the
Hamming distance between two successive word of
$\epsilon(\beta)$ (which is 1), the following proposition holds.

\begin{proposition}\label{epsilon}
For any $q\geq 3$ and binary word $\beta$, the list $\epsilon(\beta)$ is a Gray code.
\end{proposition}

\section{Trace partitioned Gray code for $S_{n,q}^{(k)}$}\label{Vajno}

Our construction of a Gray code for the set $S_{n,q}^{(k)}$
of cross-bifix-free words is based on two other lists:

\begin{itemize}
\item $\mathcal{F}_n^{(k)}$, a Gray code for the set of binary words
      of length $n$ avoiding $k$ consecutive $0$'s, and
\item $\mathcal H_{n,q}^{(k)}$, a Gray code for the set of $q$-ary words
      of length $n$ which begin and end by a
      non zero value and avoiding $k$ consecutive $0$'s.
      In particular,  $\mathcal H_{n,2}^{(k)}=1\cdot \mathcal F_{n-2}^{(k)}\cdot 1$.
\end{itemize}

Finally, we will define the Gray code list  $\mathcal S_{n,q}^{(k)}$ for the set
$S_{n,q}^{(k)}$ as $0^k\cdot\mathcal H_{n-k,q}^{(k)}$.

\subsection{The list $\mathcal{F}_n^{(k)}$}

Let $\mathcal C_n$ be the list of binary words defined as:
\begin{equation}
\label{bin_G}
\mathcal C_n=
\left\{
\begin{array}{cr}
\lambda&\quad\mbox{if}\quad n=0,\\
\\
1\cdot\overline{\mathcal C_{n-1}}\circ 0\cdot\mathcal C_{n-1}&\quad\mbox{if}\quad n\geq 1,

\end{array}
\right.
\end{equation}
with $\lambda$ the empty word. The list $\mathcal C_n$ is a
Gray code for the set $\{0,1\}^n$ and it is a slight
modification of the original Binary Reflected Gray Code list defined in \cite{G}.

\medskip
By the definition of $\mathcal C_n$ given in relation (\ref{bin_G}), we have for $n\geq 1$,
\begin{itemize}
\item
$\last(\mathcal C_n)=0\cdot\last(\mathcal C_{n-1})=0^n$;
\item
$\first(\mathcal C_n)=1\cdot\first(\overline{\mathcal C_{n-1}})=
1\cdot\last(\mathcal C_{n-1})=10^{n-1}$.
\end{itemize}

\medskip

Let now define the list $\mathcal{F}_n^{(k)}$ of length $n$ binary words as:
\begin{equation}
\label{f_nk}
\mathcal{F}_n^{(k)}=\left\{
\begin{array}{cr}
\mathcal C_n &\ \mathrm{if}\ 0\leq n<k,\\
\\
1\cdot\overline{\mathcal F_{n-1}^{(k)}}\circ 01\cdot\overline{\mathcal F_{n-2}^{(k)}}
\circ 001\cdot\overline{\mathcal F_{n-3}^{(k)}}\circ\cdots\circ 0^{k-1}1\cdot
\overline{\mathcal F_{n-k}^{(k)}}& \ \mathrm{if}\ n\geq k.
\end{array}
\right.
\end{equation}

For $k\geq 2$ and $n\geq 0$, $\mathcal{F}_n^{(k)}$ is a list for the set of
length $n$ binary words with no $k$ consecutive~$0$'s, and
Proposition \ref{Grayness_F} says that it is a Gray code (actually,
$\mathcal{F}_n^{(k)}$ is a adaptation of a similar list defined earlier
 \cite{V1}).

It is easy to see that the number of binary words in $\mathcal F_n^{(k)}$ is
given by $f_n^{(k)}$, the well known $k$-Fibonacci integer sequence defined by:
$$
f_n^{(k)}=
\left\{
\begin{array}{cl}
2^n & \mbox{ if}\ 0\leq n< k,\\
\\
f_{n-1}^{(k)}+f_{n-2}^{(k)}+\cdots+f_{n-k}^{(k)}, & \mbox{ if}\ n\geq k,
\end{array}
\right.
$$
and the words in $\mathcal F_n^{(k)}$
are said \emph{$k$-generalized Fibonacci words}.
For example, the list $\mathcal{F}_3^{(3)}$ for the length $3$ binary
words avoiding $3$ consecutive $0$'s is
$$\mathcal F_3^{(3)}=~(100,101,111,110,010,011,001).$$


\begin{proposition} $ $
\label{last_first}
\begin{itemize}
\item $\first(\mathcal F_n^{(k)})$ is the length $n$
prefix of the infinite periodic word $(10^{k-1}1)(10^{k-1}1)\ldots$;
\item $\last(\mathcal F_n^{(k)})$ is the length $n$
prefix of the infinite periodic word $(0^{k-1}11)(0^{k-1}11)\ldots$.
\end{itemize}
\end{proposition}
\begin{proof}
For the first point,
if $1\leq n<k$, then $\first(\mathcal F_n^{(k)})=\first(\mathcal C_n)=10^{n-1}$; and
if $n=k$, then $\first(\mathcal F_n^{(k)})=
1\cdot\first(\overline{\mathcal F_{n-1}^{(k)}})=
1\cdot\last(\mathcal C_{n-1})=10^{k-1}$, and the statement holds
in both cases.

Now, if $n> k$, by the definition of $\mathcal F_n^{(k)}$ we have
\begin{eqnarray*}
\first(\mathcal F_n^{(k)}) & = & 1\cdot\first(\overline{\mathcal F_{n-1}^{(k)}})\\
                                & = & 1\cdot\last(\mathcal F_{n-1}^{(k)})\\
                                & = & 10^{k-1}1\cdot\last(\overline{\mathcal F_{n-k-1}^{(k)}})\\				
                                & = & 10^{k-1}1\cdot\first(\mathcal F_{n-k-1}^{(k)}),				
\end{eqnarray*}
and recursion on $n$ completes the proof.

For the second point,
if $1\leq n<k$, then $\last(\mathcal F_n^{(k)})=\last(\mathcal C_n)=0^n$; and
if $n=k$, then $\last(\mathcal F_n^{(k)})=
0^{k-1}1$, and the statement holds
in both cases.

Now, if $n> k$, we have
\begin{eqnarray*}
\last(\mathcal F_n^{(k)}) & = & 0^{k-1}1\cdot\last(\overline{\mathcal F_{n-k}^{(k)}})\\
                               & = & 0^{k-1}1\cdot\first(\mathcal F_{n-k}^{(k)}),				
\end{eqnarray*}
and by the first point of the present proposition, recursion on $n$ completes the proof.
\end{proof}

\begin{proposition}
\label{Grayness_F}
The list $\mathcal F_n^{(k)}$ is a Gray code where
two consecutive strings differ in a single position.
\end{proposition}
\begin{proof}
It is enough to prove that there is a `smooth' transition
between any two consecutive lists in the definition
of $\mathcal F_n^{(k)}$ given in relation (\ref{f_nk}),
that is, for any $\ell$, $1\leq \ell\leq k-1$, the words

$$\alpha=0^{\ell-1}1\cdot \last(\overline{\mathcal F_{n-\ell}^{(k)}})=
0^{\ell-1}1\cdot \first(\mathcal F_{n-\ell}^{(k)})$$
and
$$\beta=0^{\ell}1\cdot \first(\overline{\mathcal F_{n-\ell-1}^{(k)}})=
0^{\ell}1\cdot \last(\mathcal F_{n-\ell-1}^{(k)})
$$
differ in a single position.
By Proposition \ref{last_first},
$$\alpha=0^{\ell-1}1\alpha'$$
and
$$\beta=0^\ell1\beta'$$
with $\alpha'$ and $\beta'$ appropriate length prefixes of
$(10^{k-1}1)(10^{k-1}1)\ldots$ and
$(0^{k-1}11)(0^{k-1}11)\ldots$, and so $\alpha$ and $\beta$
differ only in position $\ell$.
\end{proof}

As a by-product of the proof of the previous proposition we have the following remark
which is critical in algorithm {\tt process} used for the generating algorithm
in Section \ref{gen_S}.

\begin{Rem}
\label{the_rem}
If $\alpha=a_1a_2\ldots a_n$ and $\beta=b_1b_2\ldots b_n$
are two successive words in $\mathcal F_n^{(k)}$ which differ in
position $\ell$, then either $\ell=n$ or $a_{\ell+1}=b_{\ell+1}=1$.
\end{Rem}

\subsection{The list $\mathcal H_{n,q}^{(k)}$}

Let $\mathcal H_{n,q}^{(k)}$ be the list defined by:

\begin{equation}
\mathcal H_{n,q}^{(k)}=
\epsilon(\alpha_1)\circ
\overline{\epsilon(\alpha_2)}\circ
\epsilon(\alpha_3)\circ
\overline{\epsilon(\alpha_4)}\circ \cdots\circ
\epsilon'(\alpha_{f_{n-2}^{(k)}})
\end{equation}
with $\alpha_i=1\phi_i1$ and $\phi_i$ is the
$i$-th binary word in the list $\mathcal F_{n-2}^{(k)}$, and
$\epsilon'(\alpha_{f_{n-2}^{(k)}})$ is
$\epsilon(\alpha_{f_{n-2}^{(k)}})$ or
$\overline{\epsilon(\alpha_{f_{n-2}^{(k)}})}$
according on whether $f_{n-2}^{(k)}$ is odd or even.

Clearly, $\mathcal H_{n,q}^{(k)}$ is a list
for the set of $q$-ary words of length $n$ which begin and end by a
non zero value, and with no $k$ consecutive $0$'s.
In particular,  $\mathcal H_{n,2}^{(k)}=1\cdot \mathcal F_{n-2}^{(k)}\cdot 1$.

\begin{proposition}
The list $\mathcal H_{n,q}^{(k)}$ is a Gray code.
\label{Gray_H}
\end{proposition}
\begin{proof}
From Proposition \ref{epsilon} it follows
that consecutive words in each list $\epsilon(\alpha_i)$ and
$\overline{\epsilon(\alpha_i)}$ differ in a single position
(and by $+1$ or $-1$ in this position).
To prove the statement it is enough to show that,
for two consecutive binary words $\phi_i$ and $\phi_{i+1}$
in $\mathcal F_{n-2}^{(k)}$,
both pair of words
\begin{itemize}
\item $\last(\epsilon(1\phi_i 1))$
and  $\first(\overline{\epsilon(1\phi_{i+1} 1)})=
\last(\epsilon(1\phi_{i+1} 1))$, and
\item $\last(\overline{\epsilon(1\phi_i 1)})=
\first(\epsilon(1\phi_i 1))$
and  ${\rm first}(\epsilon(1\phi_{i+1} 1))$
\end{itemize}
differ in a single position.

In the first case, by Proposition \ref{first_last}, the first symbols of
$\last(\epsilon(1\phi_i 1))$ and of
$\last(\epsilon(1\phi_{i+1} 1))$
are both $(q-1)$, and the other symbols are either $1$ if
$q$ is odd, or $(q-1)$ if $q$ is even; and since
$\phi_i$ and $\phi_{i+1}$ differ in a single position, the
result holds.

In the second case, $\first(\epsilon(1\phi_i 1))=1\phi_i 1$
and $\first(\epsilon(1\phi_{i+1} 1))=1\phi_{i+1} 1$,
and again the result holds.
\end{proof}

\subsection{The list $\mathcal S_{n,q}^{(k)}$}

Now we define the list $\mathcal S_{n,q}^{(k)}$ as

$$
\mathcal S_{n,q}^{(k)}=0^k\cdot \mathcal H_{n-k,q}^{(k)},
$$
and clearly, $\mathcal S_{n,q}^{(k)}$ is a list for the set
of cross-bifix-free words $S_{n,q}^{(k)}$.
In particular,

$$
\mathcal S_{n,2}^{(k)}=0^k1\cdot \mathcal F_{n-k-2}^{(k)}\cdot 1,
$$
for example, the set $\mathcal{S}_{8,2}^{(3)}$ of length $8$
binary cross-bifix-free words which begin by $000$ is

$$\mathcal{S}_{8,2}^{(3)}= 000 1 \cdot \mathcal{F}_{3}^{(3)} \cdot 1 =$$
$$=(00011001,00011011,00011111,00011101,00010101,00010111,00010011).$$

A consequence of Proposition \ref{Gray_H} is the next proposition.

\begin{proposition}\label{cross-gray}
The list $\mathcal{S}_{n,q}^{(k)}$ is a Gray code.
\end{proposition}

For the sake of clearness, we illustrate the previous construction for the Gray code list
$\mathcal{S}_{8,3}^{(3)}$ on the alphabet $A=\{0,1,2\}$.
We have:
$$
\begin{array}{rcl}
%
\mathcal G_{3,3}&=&(111,112,122,121,221,222,212,211);\\
\\
\mathcal G_{4,3}&=&(1111,1112,1122,1121,1221,1222,1212,1211,2211,2212,2222,\\
            & &  \ 2221,2121,2122,2112,2111);\\
\\
\mathcal G_{5,3}&=&(11111,\ldots,12111,22111,\ldots,21111);
\end{array}
$$
and
$$\begin{array}{rcl}
\mathcal{S}_{8,3}^{(3)}&=&(00011001,00011002,00012002,00012001,00022001,00022002,\\
& & \ 00021002,00021001,00021011,\ldots,00011011,00011111,\ldots\\
& & \ \ldots,00021111,00021101,\ldots,00011101,00010101,00010102,\\
& & \ 00010202,00010201,00020201,00020202,00020102,00020101,\\
& & \ 00020111,\ldots,00010111,00010011,00010012,00010022,\\
& & \ 00010021,00020021,00020022,00020012,00020011).

\end{array}
$$

\section{Algorithmic considerations}

In this section we give a generating algorithm for
binary words in the list $\mathcal F_n^{(k)}$
and an algorithm expanding binary words; then, combining
them, we obtain a generating algorithm for the list $\mathcal H_{n,q}^{(k)}$,
and finally prepending $0^k$ to each word in $\mathcal H_{n-k,q}^{(k)}$
the list $\mathcal S_{n,q}^{(k)}$ is obtained.
The given algorithms are shown to be efficient.

\medskip

The list $\mathcal{F}_n^{(k)}$ defined in (\ref{f_nk})
has not a straightforward algorithmic implementation, and
now we explain how $\mathcal{F}_n^{(k)}$ can be defined
recursively as the concatenation
of at most two lists, then we will give a generating algorithm
for it. Let $\mathcal F_n^{(k)}(u)$, $0\leq u\leq k-1$, be the sublist
of $\mathcal F_n^{(k)}$ formed by strings beginning by at most
$u$ $0$'s. By the definition of $\mathcal{F}_n^{(k)}$, it follows that
$\mathcal{F}_n^{(k)}=\mathcal{F}_n^{(k)}(k-1)$, and

\begin{eqnarray*}
\mathcal{F}_n^{(k)}(0) & = & 1\cdot\overline{\mathcal F_{n-1}^{(k)}}\\
                       & = & 1\cdot\overline{\mathcal F_{n-1}^{(k)}(k-1)},
\end{eqnarray*}
and for $u>0$,

\begin{eqnarray*}
\mathcal{F}_n^{(k)}(u) & = &
1\cdot\overline{\mathcal F_{n-1}^{(k)}}\circ 01\cdot\overline{\mathcal F_{n-2}^{(k)}}\circ
\cdots\circ 0^u1\cdot
\overline{\mathcal F_{n-u-1}^{(k)}}\\
& = & 1\cdot\overline{\mathcal F_{n-1}^{(k)}}\circ 0\cdot (1\cdot\overline{\mathcal F_{n-2}^{(k)}}\circ
\cdots\circ 0^{u-1}1\cdot
\overline{\mathcal F_{n-u-1}^{(k)}})\\
& = & 1\cdot\overline{\mathcal F_{n-1}^{(k)}}\circ 0\cdot
 \mathcal{F}_{n-1}^{(k)}(u-1).
\end{eqnarray*}

By the above considerations we have the following proposition.

\begin{proposition}
\label{Pr_Fnku}
Let $k\geq 2$, $0\leq u\leq k-1$,
and $\mathcal F_n^{(k)}(u)$ be the list defined as:
\begin{equation}
\label{Fnku}
\mathcal F_n^{(k)}(u)=
\left\{
\begin {array}{ccc}
\lambda  & {\rm if} & n=0,\\ \\
1\cdot \overline{\mathcal F_{n-1}^{(k)}(k-1)}
& {\rm if} & n>0 {\ \rm and\ } u=0,\\ \\
1\cdot \overline{\mathcal F_{n-1}^{(k)}(k-1)}\,
\circ
0\cdot \mathcal F_{n-1}^{(k)}(u-1)
& {\rm if} & n,u>0.\\
\end {array}
\right.
\end{equation}
Then $\mathcal F_n^{(k)}(k-1)$ is the list
$\mathcal F_n^{(k)}$ defined by relation (\ref{f_nk}).
\end{proposition}

Now we explain how the relation (\ref{Fnku}) defining the list
$\mathcal F_n^{(k)}(u)$ can be implemented in a generating algorithm.
It is easy to check that
$\mathcal F_n^{(k)}=\mathcal F_n^{(k)}(k-1)$ has the following properties:
for $\alpha=a_1a_2\ldots a_n$ and $\beta=b_1b_2\ldots b_n$
two consecutive binary words in $\mathcal F_n^{(k)}$,
there is a $p$ such that
\begin{itemize}
\item $a_i=b_i$ for all $i$, $1\leq i\leq n$, except
      $b_p=1-a_p$,
\item $0^{k-1}$ can not be a suffix of $a_1a_2\ldots a_{p-1}=b_1b_2\ldots b_{p-1}$,
\item the sublist of $\mathcal F_n^{(k)}$ formed by the strings with the prefix
      $b_1b_2\ldots b_p$ is $b_1b_2\ldots b_p\cdot \mathcal L$,
      where $\mathcal L$ is $\mathcal F_{n-p}^{(k)}(u-1)$ or $\overline{\mathcal F_{n-p}^{(k)}(u-1)}$
      according to the prefix  $b_1b_2\ldots b_p$ has an even or odd number of $1$'s, and
      $u$ is equal to $k$ minus the length of the maximal $0$ suffix of $b_1b_2\ldots b_p$.
\end{itemize}

Let us consider procedure {\tt gen\_fib} in Figure \ref{algo} where {\tt process} switches
the value of $b[pos]$ (that is, $b[pos]:=1-b[pos]$), and prints the
obtained binary string $b$.
By the above remarks and relation (\ref{Fnku}) in
Proposition \ref{Pr_Fnku} it follows that after the initialization of
$b$ by the first string in $\mathcal F_m^{(k)}$ (given in Proposition
\ref{last_first}) and printing
it out, the call of {\tt gen\_fib($1$,$k-1$,$0$)} produces the list
$\mathcal F_m^{(k)}$.
Moreover, as we will show below, for $m=n-1$ and after the appropriate initialization of $b=b_1b_2\dots b_n$
the call of {\tt gen\_fib($k+2$,$k-1$,$0$)} produces the list
$0^k1\cdot \mathcal F_{n-k-2}^{(k)}\cdot 1=\mathcal S_{n,2}^{(k)}$.

Procedure {\tt gen\_fib} is an efficient generating  procedure.
Indeed, each recursive call induced by {\tt gen\_fib} is either
\begin{itemize}
\item a terminal call (which does not produce other calls), or
\item a call producing two recursive calls, or
\item a call producing one recursive call, which in turn
      is in one of the previous two cases.
\end{itemize}
By {\it `CAT'} principle in \cite{Rb} it follows that  procedure {\tt gen\_fib}
runs in constant amortised time.

\begin{figure}[h]
\begin{center}
\begin{tabular}{|c|}
\hline
{\tt
\begin{minipage}[c]{.46\linewidth}
\begin{tabbing}
\hspace{1.cm}\=\hspace{0.9cm}\=\hspace{1.0cm}\=\hspace{2.9cm}
            \=\hspace{2.9cm}\kill
procedure gen\_fib($pos$,$u$,$dir$)\\
global $b$,$k$,$m$;\\
if $pos\leq m$\\
then \> if $u=0$ \\
     \> then gen\_fib($pos+1$,$k-1$,$1-dir$);\\
     \> else \>  if $dir=0$  \\
     \>      \>  then \> gen\_fib($pos+1$,$k-1$,$1$);\\
     \>      \>       \> process($pos$);\\
     \>      \>       \> gen\_fib($pos+1$,$u-1$,$0$);\\
     \>      \>  else \> gen\_fib($pos+1$,$u-1$,$1$);\\
     \>      \>       \> process($pos$);\\
     \>      \>       \> gen\_fib($pos+1$,$k-1$,$0$);\\
     \>      \>  end if \\
     \> end if \\
end if \\
end procedure.
\end{tabbing}
\end{minipage}
}\\
\hline
\end{tabular}
\caption{Algorithm producing the list $\mathcal F_n^{(k)}$
or $\mathcal S_{n,q}^{(k)}$, according to the
initial values of $m$, $b$ and the
definition of {\tt process} procedure.
\label{algo}}
\end{center}
\end{figure}

\subsection{Generating $\mathcal S_{n,2}^{(k)}$}

After the initialization of $b_1b_2\ldots b_n$ by
$0^k1\cdot \first(\mathcal F_{n-k-2}^{(k)})\cdot 1$, with
$\first(\mathcal F_{n-k-2}^{(k)})$ given in Proposition
\ref{last_first}, and printing it out, the call of {\tt gen\_fib($k+2$,$k-1$,$0$)} where
\begin{itemize}
\item $m=n-1$, and
\item procedure {\tt process} called by  {\tt gen\_fib}
      switches the value of $b[pos]$ (that is, $b[pos]:=1-b[pos]$)
      and prints $b$
\end{itemize}
produces, in constant amortized time,  the list
$0^k1\cdot \mathcal F_{n-k-2}^{(k)}\cdot 1=0^k\cdot \mathcal H_{n-k,2}^{(k)}$
which is, as mentioned before, the list $\mathcal S_{n,2}^{(k)}$.

\subsection{Generating $\mathcal S_{n,q}^{(k)}$, $q>2$}
\label{gen_S}

Before discussing the expansion algorithm {\tt expand}
needed to produce the list $\mathcal S_{n,q}^{(k)}$ when $q>2$
we show that {\tt gen\_tuple} procedure in Figure \ref{algo_qary}, on which {\tt expand} is
based, is an efficient generating algorithm for the list $\mathcal{G}_{n,q}$
defined in relation (\ref{g_nq}).
Procedure {\tt gen\_tuple} is a `naive' odometer principle based algorithm, see again \cite{Rb},
and we have the next proposition.

\begin{proposition}
\label{eff_en_tuple}
After the initialization of $v$ by $11\cdots 1$,
the first word in $\mathcal{G}_{n,q}$, and $d_i$ by $1$,
for $1\leq i\leq n$, procedure {\tt gen\_tuple}
produces the list $\mathcal{G}_{n,q}$ in constant amortized time.
\end{proposition}
\begin{proof}
The total amount of computation of {\tt gen\_tuple}
is proportional to the number of times the statement $i:=i-1$
is performed in the inner {\tt while} loop; and for a given $q$ and $n$
let denote by $c_n$ this number.
So, the average complexity (per generated word) of {\tt gen\_tuple}
is $\frac{c_n}{q^n}$.
Clearly, $c_1=q-1$ and $c_n=(q-1)\cdot n+q\cdot c_{n-1}$,
and a simple recursion shows that $c_n=q \cdot \frac{q^n-1}{q-1}-n$ and
finally the average complexity of {\tt gen\_tuple} is $\frac{c_n}{q^n}\leq \frac{q}{q-1}$.
\end{proof}

\begin{figure}[h]
\begin{center}
\begin{tabular}{|c|}
\hline
{\tt
\begin{minipage}[c]{.46\linewidth}
\begin{tabbing}
\hspace{0.6cm}\=\hspace{1.2cm}\=\hspace{1.0cm}\=\hspace{2.9cm}
            \=\hspace{2.9cm}\kill
procedure gen\_tuple\\
global $v$,$d$,$n$;\\
output $v$;\\
do \> $i:=n$;\\
   \> while \> $i\geq 1$ and\\
   \>       \> ($v[i]=q-1$ and $d[i]=1$ or $v[i]=1$ and $d[i]=-1$)\\
   \>       \> $d[i]:=-d[i]$;\\
   \>       \> $i:=i-1$;\\
   \> end while\\
   \> if $i\geq 1$ then $v[i]:=v[i]+d[i]$; output $v$; end if \\
while $i\geq 1$\\
end procedure.
\end{tabbing}
\end{minipage}
}\\
\hline
\end{tabular}
\caption{Odometer algorithm producing the list
$\mathcal{G}_{n,q}$.
\label{algo_qary}}
\end{center}
\end{figure}

Now we adapt algorithm {\tt gen\_tuple} in order to obtain
procedure {\tt expand} producing the expansion of a words;
and like  {\tt gen\_tuple}, procedure {\tt expand} has a constant average time complexity.
More precisely, for a words $b=b_1b_2\dots b_n$ in $\{0,1,\ldots ,q\}^n$,
with $b_{\ell+1},b_n\neq 0$
let $b'$ denote the {\it trace} of $b_{\ell+1}b_{\ell+2}\dots b_n$, that is, the word obtained
from $b_{\ell+1}b_{\ell+2}\dots b_n$ by replacing each non-zero value by $1$, and $b''$ that
obtained by erasing each $0$ letter in $b_{\ell+1}b_{\ell+2}\dots b_n$. Procedure {\tt expand}
produces the list:

\begin{itemize}
\item $b_1b_2\dots b_{\ell}\cdot\epsilon(b')$ if the initial value of $b$ is such that $b''$
is the first word in $\mathcal{G}_{|b''|,q}$, or
\item  $b_1b_2\dots b_{\ell}\cdot\overline{\epsilon(b')}$ if the initial value of $b$ is such that $b''$
is the last word in $\mathcal{G}_{|b''|,q}$.
\end{itemize}
The initial value of $d_{\ell+1},d_{\ell+2},\dots ,d_n$ are given by:
if $b_i=1$, then $d_i=1$; and
if $b_i=q-1$, then $d_i=-1$;
otherwise $d_i$ is not defined.
In order to access in constant time from a position $i$ in the current
word $b$, with  $b_i\neq 0$,  to the previous one, additional data structures are used.
The array $prec$ is defined by:
if $b_i\neq 0$, then $prec_i=j$, where
$j$ is the rightmost position in $b$, at the left of $i$ and with $b_j\neq 0$;
and for convenience $prec_i=0$ if $i$ is the leftmost non-zero position in $b$.

\begin{figure}[h]
\begin{center}
\begin{tabular}{|c|}
\hline
{\tt
\begin{minipage}[c]{.46\linewidth}
\begin{tabbing}
\hspace{0.6cm}\=\hspace{1.2cm}\=\hspace{1.0cm}\=\hspace{2.9cm}
            \=\hspace{2.9cm}\kill
procedure expand\\
global $b$,$d$,$\ell$,$n$,$prec$;\\
output $v$;\\
do \> $i:=n$;\\
   \> while \> $i\geq\ell+1$ and\\
   \>       \> ($b[i]=q-1$ and $d[i]=1$ or $b[i]=1$ and $d[i]=-1$)\\
   \>       \> $d[i]:=-d[i]$;\\
   \>       \> $i:=prec[i]$;\\
   \> end while\\
   \> if $i\geq\ell+1$ then $b[i]:=b[i]+d[i]$; output $b$; end if \\
while $i\geq\ell+1$\\
end procedure.
\end{tabbing}
\end{minipage}
}\\ \hline
\end{tabular}
\caption{Algorithm expanding a word $b$
and mimicking procedure {\tt gen\_tuple}.
\label{algo_expand}}
\end{center}
\end{figure}

\begin{figure}[h]
\begin{center}
\begin{tabular}{|c|}
\hline
{\tt
\begin{minipage}[c]{.46\linewidth}
\begin{tabbing}
\hspace{0.9cm}\=\hspace{1.2cm}\=\hspace{1.0cm}\=\hspace{2.9cm}
            \=\hspace{2.9cm}\kill
procedure process($pos$)\\
global $b$,$d$,$succ$,$prec$;\\
if $b[pos]=0$\\
then \> $a:=prec[pos+1]$; $succ[a]:=pos$; $succ[pos]:=pos+1$;\\
     \> $prec[pos]:=a$; $prec[pos+1]:=pos$;\\
     \> $b[pos]:= b[pos+1]$;\\
     \> $d[pos]:=d[pos+1]$;\\
else \>$a:=prec[pos]$; $z:=succ[pos]$;\\
     \>$prec[z]:=a$; $succ[a]:=z$;\\
     \>$b[pos]:=0$;\\
expand;\\
end procedure.
\end{tabbing}
\end{minipage}
}\\ \hline
\end{tabular}
\caption{Procedure {\tt process} called by {\tt gen\_fib}
in order to generate the list~$\mathcal S_{n,q}^{(k)}$.
\label{proc_process}}
\end{center}
\end{figure}

Now we explain procedure {\tt process}; it calls {\tt expand}
and we will show that when {\tt gen\_fib}
in turn calls procedure {\tt process} in Figure \ref{proc_process}, 
then it produces the list $\mathcal S_{n,q}^{(k)}$,
with $q>2$.
The parameter $pos$ of {\tt process} is given by the corresponding call
of {\tt gen\_fib}, and it gives the position in the current
word $b_1b_2\ldots b_n$ in $\mathcal S_{n,q}^{(k)}$ where $b_{pos}$ changes
from a non-zero value to $0$, or {\it vice versa}.
By Remark \ref{the_rem} and the definition of the list
$\mathcal S_{n,q}^{(k)}$ from $\mathcal H_{n-k,q}^{(k)}$,
and so from $\mathcal F_{n-k-2,q}^{(k)}$, it follows that
$b_{pos+1}\neq 0$.
Procedure {\tt process}, sets $b_{pos}$ to $0$ if previously $b_{pos}\neq 0$;
and sets $b_{pos}$ to $b_{pos+1}$ if previously $b_{pos}=0$,  which
according to Proposition \ref{first_last}, Remark \ref{the_rem} and the definition of the expansion
operation is the new value of $b_{pos}$.
In order to access in constant time from a non-zero position in the array
$b$ to the previous one, {\tt process} uses array $prec$ of procedure {\tt expand} 
and array $succ$, defined as:
$succ_i=j$, with $j$ the leftmost position in $b$, at the right of $i$ and with
$b_j\neq 0$, and $succ_i$ is not defined if $i$ is the rightmost non-zero position. 
In addition, procedure {\tt process} updates both
arrays $prec$ and $succ$.

For given $q>2$, $k\geq 2$ and  $n\geq k+2$, after the initialization of $b_1b_2\ldots b_n$ by
$0^k1\cdot \first(\mathcal F_{n-k-2}^{(k)})\cdot 1$,
as for generating $\mathcal S_{n,2}^{(k)}$, the call of {\tt gen\_fib($k+2$,$k-1$,$0$)} where
\begin{itemize}
\item $m=n-1$, and
\item procedure {\tt process} is that in Figure \ref{proc_process}, and
\item procedure {\tt expand} that in Figure \ref{algo_expand}, with $\ell=k+1$
\end{itemize}
produces, in constant amortized time, the list $\mathcal S_{n,q}^{(k)}$.

\section {Conclusion and further works}
The cross-bifix-free sets $S_{n,q}^{(k)}$ defined in 
\cite{singa} have the cardinality close to the optimum. They
are constituted by particular words
$s_1s_2\ldots s_n$ of length $n$ over a $q$-ary alphabet. Each word has the form $0^ks_{k+1}s_{k+2}\dots s_n$ where
$s_{k+1}$ and $s_n$ are different from $0$ and $s_{k+1}s_{k+2}\dots  s_{n-1}$ does not contain $k$ consecutive 0's.
We have provided a Gray code for $S_{n,q}^{(k)}$ by defining a Gray code for the words
$s_{k+1}s_{k+2}\dots s_n$ and then prepending the prefix $0^k$ to them. 
Moreover, an efficient generating algorithm for the obtained Gray code is given. 
We note that this Gray code is \emph{trace partitioned}
in the sense that all the words with the same trace are consecutive.
To this aim we used a Gray code for restricted binary strings \cite{V1}, opportunely replacing
the bits 1 with the symbols of the alphabet different from 0.

A future investigation could be the definition of a Gray code which is \emph{prefix partitioned}, where
all the words with the same prefix are consecutive. Actually,
the definition of the sets $S_{n,q}^{(k)}$ shows that it is sufficient to define
a prefix partitioned Gray code for the subwords $s_{k+1}s_{k+2}\dots s_n$.

An interesting question arising when one deals with a Gray code $\mathcal L$ on a set is the possibility
to define it in such a way that the Hamming distance between 
$\last (\mathcal L)$ and $\first (\mathcal L)$ is $1$ (circular Gray code).
Usually it is not so easy to have a circular Gray code, unless the elements of the set are
not subject to constraints; in our case it is worth to study if 
the ground-set we are dealing with (which is a cross-bifix free set) 
allows to find a circular Gray code.

\end{document}